\documentclass[journal]{IEEEtran}
\usepackage{flushend}
\usepackage[boxruled]{algorithm2e}
\usepackage{cite}
\usepackage{graphicx}
\usepackage{psfrag}
\usepackage{subfigure}
\usepackage{url}
\usepackage{amsmath}
\usepackage{array}
\usepackage{amssymb}
\usepackage{amsfonts}
\newtheorem{proposition}{Proposition}
\newtheorem{lemma}{Lemma}

\newtheorem{remark}{Remark}
\newtheorem{definition}{Definition}

\newtheorem{example}{Example}

\title{Network-Error Correcting Codes using Small Fields}
\begin{document}
\author{
\authorblockN{K.~Prasad and B.~Sundar Rajan,\\}
\authorblockA{Dept. of ECE, IISc, Bangalore 560012, India\\
Email: \{prasadk5,bsrajan\}@ece.iisc.ernet.in\\}
}
\maketitle
\thispagestyle{empty}	
\begin{abstract}
Existing construction algorithms of block network-error correcting codes require a rather large field size, which grows with the size of the network and the number of sinks, and thereby can be prohibitive in large networks. In this work, we give an algorithm which, starting from a given network-error correcting code, can obtain another network code using a small field, with the same error correcting capability as the original code. An algorithm for designing network codes using small field sizes proposed recently by Ebrahimi and Fragouli can be seen as a special case of our algorithm. The major step in our algorithm is to find a least degree irreducible polynomial which is coprime to another large degree polynomial. We utilize the algebraic properties of finite fields to implement this step so that it becomes much faster than the brute-force method. As a result the algorithm given by Ebrahimi and Fragouli is also quickened. 
\end{abstract}
\let\thefootnote\relax\footnotetext
{
Part of the content of this work was presented at ISIT 2011 held at St. Petersburg, Russia during July 31 - Aug. 5, 2011. 
}
\section{Introduction}
\label{sec1}
Network coding was introduced in \cite{ACLY} as a means to improve the rate of transmission in networks. Linear network coding was introduced in \cite{CLY}. Deterministic algorithms exist \cite{KoM,JSCEEJT,Har} to construct \textit{scalar network codes} (in which the input symbols and the network coding coefficients are scalars from a finite field) which achieve the maxflow-mincut capacity in the case of acyclic networks with a single source which wishes to multicast a set of finite field symbols to a set of $N$ sinks, as long as the field size $q>N$. Finding the minimum field size over which a network code exists for a given network is known to be NP hard \cite{LeL}. An algorithm was proposed in \cite{EbF} which attempts to find network codes using small field sizes, given a network coding solution for the network over some larger field size $q>N.$ The algorithms of \cite{EbF} also apply to linear deterministic networks \cite{ADT}, and for \textit{vector network codes} (where the source seeks to multicast a set of vectors, rather than just finite field symbols). In this work, we are explicitly concerned about the scalar network coding problem, although the same techniques can be easily extended to accommodate for vector network coding and linear deterministic networks, if permissible, as in the case of \cite{EbF}.

Network-error correction, which involved a trade-off between the rate of transmission and the number of correctable network-edge errors, was introduced in \cite{YeC} as an extension of classical error correction to a network setting. Along with subsequent works \cite{Zha} and \cite{YaY}, this generalized the classical notions of the Hamming weight, Hamming distance, minimum distance and various classical error control coding bounds to their network counterparts. Algorithms for constructing network-error correcting codes which meet a generalization of the classical Singleton bound for networks can be found in \cite{Zha,YaY,Mat,GFZ}. Using the algorithm of \cite{Mat}, a network code which can correct any errors occurring in at most $\alpha$ edges can be constructed, as long as the field size $q$ is such that 
\[
q>N
\left(
\begin{array}{c}
|\cal{E}| \\
2\alpha
\end{array}
\right),
\]
where $\cal E$ is the set of edges in the network. The algorithms of \cite{Zha,YaY} have similar requirements to construct such network-error correcting codes.  This can be prohibitive when $|\cal E|$ is large, as the sink nodes and the coding nodes of the network have to perform operations over this large field, possibly increasing the overall delay in communication. In \cite{GFZ}, the bound on the field size was further tightened. However, this bound in \cite{GFZ} too potentially grows with the size of the network.

In this work, we propose an algorithm for block network-error correction using small fields. We shall restrict our algorithms and analysis to fields with binary characteristic. The techniques presented can be extended to finite fields of other characteristics without much difficultly. The contributions of this work are as follows.
\begin{itemize}
\item We propose an algorithm to construct network-error correcting codes using small fields, by first designing a network-error correcting code over a large field size using known techniques (for example, \cite{Mat}) and then using algebraic techniques to obtain a network-error correcting code over a smaller field size. The network coding version of this algorithm reduces to the algorithm proposed by Ebrahimi and Fragouli in \cite{EbF}, which we shall refer to as the EF algorithm henceforth. 
\item The major step in our algorithm is to compute a polynomial of least degree coprime with a polynomial, $f(X),$ of possibly large degree.  While it is shown in \cite{EbF} that this can be done in polynomial time, the complexity can still be large. Optimizing based on our requirement, we propose an alternate faster algorithm for computing the polynomial coprime with $f(X).$ This reduces the complexity of the EF algorithm also, which simply adopts a brute force method to do the same.
\item Illustrative examples are shown which indicate that parameters such as the initial network-error correcting code and the choice of representation of the initial large finite field influence the ability of our algorithm to obtain a network-error correcting code over a small field size.
\end{itemize}

The rest of this paper is organized as follows. In Section \ref{sec2}, we give the basic notations and definitions related to network coding, required for our purpose. Also, we review the EF algorithm briefly in Section \ref{sec2}. Section \ref{sec4} presents our algorithm for constructing network-error correcting codes using small field sizes, along with calculations of the complexity of the algorithm.  In Section \ref{sec4}, we also propose a fast way to compute the major step of our algorithm, which is to obtain a least degree polynomial coprime with another polynomial of larger degree. We also show that this fast technique reduces the running time of the EF algorithm. Examples illustrating our algorithm for network coding and error correction are presented in Section \ref{sec5}. Finally, we conclude the paper in Section \ref{sec6} with comments and directions for further research. 

\section{Preliminaries and Background}
\label{sec2}
The model for acyclic networks considered in this paper is as in \cite{CLYZ}. An acyclic network can be represented as an acyclic directed multi-graph ${\cal G} = ({\cal V}, {\cal E})$, where $\cal V$ is the set of all nodes and $\cal E$ is the set of all edges in the network. We assume that every edge in ${\cal G}$ can carry at most one symbol from a finite field $\mathbb{F}_q$. Network links with capacities greater than unity are modeled as parallel edges. The network is assumed to be instantaneous, i.e., all nodes process the same \emph {generation} (the set of symbols generated at the source at a particular time instant) of input symbols to the network in a given coding order (ancestral order \cite{CLYZ}). For an edge $e,$ let $tail(e)$ and $head(e)$ denote the start node and the end node of $e.$ An ancestral ordering can be assumed on $\cal E$ as the network is acyclic. Let $s\in\cal V$ be the source node and $\cal T$ be the set of $N(=|{\cal T}|)$ receivers. Let $h_{_T}$ be the unicast capacity for a sink node  $T\in{\cal T}$, i.e., the maximum number of edge-disjoint paths from $s$ to $T$. Then $h = \min_{T\in{\cal T}}h_{_T}$ is the max-flow min-cut capacity of the multicast connection. 

A $h'$-dimensional network code ($h'\leq h$) is one which can be used to transmit $h'$ symbols simultaneously from $s$ to all sinks $T\in {\cal T},$ and can be described \cite{KoM} by the following matrices, each having elements from the finite field $\mathbb{F}_q$. 
\begin{itemize}
\item A matrix $A$ (of size $h' \times |{\cal E}|$), which describes the way the source maps symbols onto the network. The entries of $A$ are defined as
\[
A_{i,j}=\left\{
\begin{array}{cc}
\alpha_{i,e_j} & \text{   if } s = tail(e_j),\\
0  & \text{ otherwise},
\end{array}
\right. 
\]
where $\alpha_{i,e_j} \in \mathbb{F}_q$ is the network coding coefficient at the source coupling input $i$ with edge $e_j.$
\item A matrix $K$ (of size $|{\cal E}| \times |{\cal E}|$), which describes how the symbols are processed between the edges of the network. The entries of $K$ are defined as
\[
K_{i,j}=\left\{
\begin{array}{cc}
\beta_{i,j} & \text{   if } head(e_i) = tail(e_j), \\
0  & \text{ otherwise}, 
\end{array}
\right. 
\]
where $\beta_{i,j} \in \mathbb{F}_q$ is the local encoding kernel coefficient  between $e_i$ and $e_j$.
\item $D_T$ (of size $|{\cal E}| \times h'$ for every sink $T \in \cal T$), which describes how the symbols received by the sink $T$ are processed. The entries of the matrix $D_T$ are defined as  
\[
{D_T}_{i,j}=\left\{
\begin{array}{cc}
\epsilon_{e_j,i} & \text{   if } head(e_j) = T,\\
0  & \text{ otherwise}, \end{array} 
\right. 
 \]
where $\epsilon_{e_j,i} \in \mathbb{F}_q$ describes the coupling between the symbols on $e_j$ and the $i^{th}$ input.
\end{itemize}
Let $F = (I-K)^{-1},$ where $I$ is the identity matrix of size $|{\cal E}|.$ Note that $F$ is well defined as $(I-K)$ is an invertible matrix, as $K$ is strictly upper-triangular. We then have the following definition.
\begin{definition}
\cite{KoM}
\label{nettransfermatrix}
\textit{The network transfer matrix}, $M_{T}$ for a $h'$-dimensional network code, corresponding to a sink node ${T} \in \cal T$ is a full rank $h' \times h'$ matrix defined as 
$
M_{T}:=AFD_{T}=AF_{T},
$
where $F_T:=FD_{T}.$
\end{definition} 	

The matrix $M_T$ governs the input-output relationship at sink $T.$ The problem of designing a $h'$-dimensional network code then implies making a choice for the matrices $A, F,$ and $D_T,$ such that the matrices $\left\{M_T : T \in {\cal T}\right\}$ have rank $h'$ each. We thus consider each element of $A, F$, and $D_T$ to be a variable $X_i$ for some positive integer $i$, which takes values from the finite field $\mathbb{F}_q.$ Let $\left\{X_i\right\}$ be the set of all variables, whose values define the network code. The variables $X_i$s are known as the \textit{local encoding coefficients} \cite{CLYZ}. For an edge $e$ in a network with a $h'$-dimensional network code in place, the \textit{global encoding vector} \cite{CLYZ} is a $h'$ dimensional vector which defines the particular linear combination of the $h'$ input symbols which flow through $e.$ It is known  \cite{KoM,JSCEEJT,Har} that deterministic methods of constructing a $h$-dimensional network code exist, as long as $q > N.$

Let $\Lambda$ be the length of the longest path from the source to any sink. Because of the structure of the matrices $A, F$
and $D_T$, it is seen \cite{EbF} that the matrix $M_T$ has degree at most $\Lambda$ in any particular variable $X_i$ and also a total degree (sum of the degrees across all variables in any monomial) of $\Lambda$. Let $f_{_T}\left(X_1,X_2, ..X_{|\left\{X_i\right\}|}\right)$ be the determinant of $M_T$ and $f(X_1,X_2, ..X_{|\left\{X_i\right\}|}) = \prod_{T\in{\cal T}}f_{_T}.$ Then the degree in any variable (and the total degree) of the polynomials $f_{_T}$ and $f$ are at most $h'\Lambda$ and $Nh'\Lambda$ respectively. 

A brief version of the EF algorithm is given in Algorithm \ref{alg:construction}. 
\begin{algorithm}
$(1)$ Assign values $\alpha_i$s to the scalar coding coefficients $X_i$s from an appropriate field $\mathbb{F}_{2^k} \left(2^k=2^{\lceil log(N) \rceil +1} > N \right)$ such that the network transfer matrices	 $M_T$s to all the sinks are invertible. 

$(2)$ Express every $X_i=\alpha_i$ as a binary polynomial $p_i(X)$ of degree at most $k-1$ using the usual polynomial representation of the finite field $\mathbb{F}_{2^k},$ for a particular choice of the primitive polynomial of degree $k.$

$(3)$ Substituting these polynomials representing the $X_i$s in the matrices $M_T,$ calculate the determinants of $M_T$ as the polynomials $f_{_T}(X) \in \mathbb{F}_2[X],$ and also find $f(X)=\prod_{T\in{\cal T}}f_{_T}(X).$ Then, $f(X)$ is non-zero and has degree at most $N(k-1)h\Lambda$ in the variable $X.$

$(4)$ Find an irreducible polynomial of least degree, $g(X),$ which is coprime with $f(X).$  

$(5)$ Let $X_i= p_i(X)(\text{mod}~g(X)).$ Thus, each $X_i$ can be viewed as an element in $\frac{\mathbb{F}_2[X]}{\left(g(X)\right)}.$ Also, for each sink $T,$ the matrices $M_T$ remain invertible as $f_{_T}(X)(\text{mod}~g(X))\neq 0,$ as $f(X)(\text{mod}~g(X))\neq 0.$ 
\caption{Scalar network coding algorithm using small fields - \cite{EbF}\vspace{0.2cm}}
\label{alg:construction}
\end{algorithm}
Note that the key step in Algorithm \ref{alg:construction} is step (4), where an irreducible polynomial $g(X)$ of least degree is to be found. It is shown in \cite{EbF} that such a coprime $g(X)$ exists and and can be computed with $O\left(n^2log(n)\right)$ operations, where $n=deg(f(X))=Nh\Lambda \lceil log(N) \rceil.$
\section{Network-error Correcting codes using small fields}
\label{sec4}
This section presents the major contribution of this work. After briefly reviewing the network-error correcting code construction algorithm in \cite{Mat}, we proceed to give an algorithm which can obtain network-error correcting codes using small finite fields.  
\subsection{Network-Error Correcting Codes - Approach of \cite{Mat}}
An edge is said to be in error if its input symbol and output symbol (both from some appropriate field $\mathbb{F}_q$) are not the same. We model the edge error as an additive error from $\mathbb{F}_q$. A \textit{network-error} is a $|\cal E|$ length vector over $\mathbb{F}_q$, whose components indicate the additive errors on the corresponding edges. A network code which enables every sink to correct any errors in any set of edges of cardinality at most $\alpha$ is said to be an $\alpha$ \textit{network-error correcting code}.  There have been different approaches to network-error correction \cite{YeC,Zha,YaY,Mat,GFZ}. We concern ourselves with the notations and approach of \cite{Mat}, as the algorithm in \cite{Mat} lends itself to be extended according to the techniques of \cite{EbF}.   

It is known \cite{YeC} that the number of messages $M$ in an $\alpha$ network-error correcting code is upper bounded according to the \textit{network Singleton bound} as 
$
M \leq q^{h-2\alpha}.
$
Assuming that the message set is a vector space over $\mathbb{F}_q$ of dimension $k,$ we have 
$
k \leq h-2\alpha.
$

A brief version of the algorithm given in \cite{Mat} for constructing an $\alpha$ network-error correcting code for a given single source, acyclic network that meets the network Singleton bound is shown in Algorithm \ref{alg:necc}. The construction of \cite{Mat} is based on the network code construction algorithm of \cite{JSCEEJT}. The algorithm constructs a network code such that all network-errors in up to $2\alpha$ edges will be corrected as long as the sinks know where the errors have occurred. Such a network code is then shown \cite{Mat} to be equivalent to an $\alpha$ network-error correcting code. Other equivalent (in terms of complexity) network-error correction algorithms can be found in \cite{Zha} \cite{YaY}.

One way to understand Algorithm \ref{alg:necc} which is relevant to our work is as follows. For each subset $F\in {\cal F}$ of $\cal E,$ Algorithm \ref{alg:necc} considers a subnetwork of the original network consisting of $k$ edge-disjoint paths from the imaginary source $s'$ to each sink $T\in {\cal T}$ and also $m_T^F$ edge-disjoint paths from $s'$ passing through the edges of $F$ to each sink $T$ which are also edge-disjoint with the $k$ paths from $s'$. On this subnetwork, Algorithm \ref{alg:necc} chooses network coding coefficients such that the $k$ information symbols can still be multicast to each sink $T$ irrespective of whatever information may flow on the $m_T^F$ paths. If the same choice of coefficients can be chosen to satisfy this multicast-like constraint for each $F \in {\cal F},$ then there is a valid $\alpha$ network-error correcting code which can be used to multicast the $k$ information symbols from the source to all sinks in the network. This understanding of Algorithm \ref{alg:necc} is key to understanding our algorithm for obtaining network-error correcting codes for small field sizes. For further details on Algorithm \ref{alg:necc}, the reader is referred to \cite{Mat}.
\begin{algorithm}
$(1)$ Let $\cal F$ be the set of all subsets of $\cal E$ of size $2\alpha.$ Add an imaginary source $s'$ and draw $k=h-2\alpha$ edges from $s'$ to $s.$

$(2)$ \ForEach{$F \in {\cal F}$}
{
$(i)$ Starting from the original network, add an imaginary node $v$ at the midpoint of each edge $e\in F$ and add an edge of unit capacity from $s'$ to each $v.$ 

$(ii)$ \ForEach{sink $T\in \cal T$} 
{
Draw as many edge disjoint paths from $s'$ to $T$ passing through the imaginary edges added at Step $(i)$ as possible. Let $m_{_T}^F (\leq 2\alpha)$ be the number of such paths. 

Draw $k$ edge disjoint paths passing through $s$ that are also edge disjoint from the $m_{_T}^F$ paths drawn in the previous step.
}

$(iii)$ Based on the techniques shown in the network coding algorithm of \cite{JSCEEJT} on the subnetwork comprising of the identified edge disjoint paths, obtain a network code with the following property. Let $B_T^F$ be the $\left(k+2\alpha\right) \times \left(k+m_{_T}^F\right)$ matrix, the columns of which are the $h$ length global encoding vectors (representing the linear combination of the $k$ input symbols and $2\alpha$ error symbols) of the incoming edges at sink $T$ corresponding to the $k+m_{_T}^F$ edge disjoint paths. Then $B_T^F$ must be full rank. As proved in \cite{Mat}, this ensures that the network code thus obtained is $\alpha$ network-error correcting and meets the network Singleton bound.
}
\caption{\vspace{0.1cm}Algorithm of \cite{Mat} for constructing a network-error correcting code that meets the network Singleton bound.\vspace{0.1cm}}
\label{alg:necc}
\end{algorithm}

It is shown in \cite{Mat} that Algorithm \ref{alg:necc} results in a network code which is an $\alpha$ network-error correcting code meeting the network Singleton bound, as long as the field size 
\begin{equation}
\label{eqn15}
q > |{\cal T}||{\cal F}|= N
\left(
\begin{array}{c}
|{\cal E}| \\
2\alpha 
\end{array}
\right).
\end{equation}
The above bound on field size was further tightened in \cite{GFZ}, where it was shown that a construction of an $\alpha$ network-error correcting code is possible if the field size $q$ is such that 
\begin{equation}
\label{eqn15a}
q > \sum_{T \in{\cal T}} |R_T(\alpha)|,
\end{equation}
where $R_T(\alpha)$ is a set defined in \cite{GFZ} for the sink $T$ in the following way.
\begin{definition}
For a sink $T,$ the set $R_T(\alpha)$ is the set of all subsets of size $2\alpha$ of the edge set ${\cal E}$ satisfying the following properties for each $\rho\in R_T(\alpha)$ .
\begin{itemize}
\item A collection of $k$ edge-disjoint paths starting from the $k$ imaginary incoming edges at the source node $s$ to sink node $T$ can be found.
\item A collection of $2\alpha$ edge-disjoint paths starting from each of the  $2\alpha$ edges to the sink $T$ in $\rho$ can be found, such that all these paths are also edge-disjoint from the $k$ paths from $s.$ 
\end{itemize}
\end{definition}
An algorithm is shown in \cite{GFZ} to construct $\alpha$ network-error correcting codes if the field size is greater than $q > \sum_{T \in{\cal T}} |R_T(\alpha)|.$ In many networks (see \cite{GFZ}, for example), this bound in (\ref{eqn15a}) could be smaller than the bound in (\ref{eqn15}). However, in this work, we use the Algorithm \ref{alg:necc} which is from \cite{Mat} rather than the algorithm from \cite{GFZ}. We shall however give the value of the bound in (\ref{eqn15a}) for an example network and show that our algorithm to obtain network-error correcting codes over small fields can obtain field sizes smaller than that of the bound in (\ref{eqn15a}) also. 

\subsection{Network-Error Correction using Small Fields - Algorithm}
Algorithm \ref{alg:necclowfieldsize} constructs a network-error correcting code using small field sizes (conditioned on the existence of an irreducible polynomial of small degree satisfying the necessary requirements indicated in Step (5) of Algorithm \ref{alg:necclowfieldsize}). Note that for the case $\alpha=0,$ Algorithm \ref{alg:necclowfieldsize} reduces to the EF algorithm, i.e., Algorithm \ref{alg:construction}.
\begin{algorithm}
$(1)$ With $q = 2^{\left\lceil log\left(N|{\cal F}|\right)\right\rceil+1} = 2^k,$ 
run Algorithm \ref{alg:necc} to find an $\alpha$ network-error correcting code meeting the network Singleton bound. Let the encoding coefficients for $X_i$ be $\alpha_i.$

$(2)$ Express every $X_i=\alpha_i$ as a binary polynomial $p_i(X)$ of degree at most $k-1$ using the usual polynomial representation of the finite field $\mathbb{F}_{2^k}.$

$(3)$ \ForEach{ $F \in {\cal F}$}
{
\ForEach{sink $T\in {\cal T}$}
{
Find a non-zero minor of the matrix $B_T^F,$ obtained from a $\left(k+m_{_T}^F\right) \times \left(k+m_{_T}^F\right)$ submatrix. At least one such minor exists as $B_T^F$ has rank = $k+m_{_T}^F.$ Let the minor be $f_{_T}^{F}(X),$ which can be of degree at most $h\Lambda log(N|{\cal F}|)$, according to Section \ref{sec2} and the choice of our field size. 
}
}

$(4)$ Calculate the polynomial 
\[
f(X)= \prod_{F\in {\cal F}}\prod_{T\in {\cal T}}f_{_T}^{F}(X),
\]
which has degree at most $N|{\cal F}|h\Lambda log(N|{\cal F}|).$

$(5)$ Find an irreducible polynomial of least degree, $g(X),$ which is coprime with $f(X).$  

$(6)$ Let $X_i= p_i(X)(\text{mod}~g(X)).$ Thus, each $X_i$ can be viewed as an element in $\frac{\mathbb{F}[X]}{\left(g(X)\right)}.$ Because of the fact that $f_{_T}^{F}(X)(\text{mod}~g(X)) \neq 0$ (as $f(X)(\text{mod}~g(X))\neq 0$), the new $B_T^F$ matrices obtained after the modulo operation are also full rank, which implies that the error correcting capability of the code is preserved. 
\caption{Network-error correcting codes under small field sizes}
\label{alg:necclowfieldsize}
\end{algorithm}
As in Algorithm \ref{alg:construction}, the major step of Algorithm \ref{alg:necclowfieldsize} is Step (5) which involves calculating a polynomial $g(X)$ coprime with a given polynomial $f(X).$ According to the complexity calculations in \cite{EbF}, a brute force computation of Step (5) would require $O(n^2log(n))$ computations, $n=deg(f(X))=N|{\cal F}|h\Lambda \lceil log(N|{\cal F}|) \rceil$. Before we propose our method to execute Step (5) efficiently in Subsection \ref{subsecfastcoprime}, we give a justification for Algorithm \ref{alg:necclowfieldsize}.
\subsection{Justification for Algorithm \ref{alg:necclowfieldsize}}
No justification is required for the steps in Algorithm \ref{alg:necclowfieldsize} except Step (5). The justification for Step (5) is as follows. Step (5) finds a $g(X)$ which is coprime with the product polynomial $f(X).$ In fact, in order to ensure that the error correction property of the original network code is preserved, it is sufficient if a polynomial $g(X)$ is coprime with each polynomial $f^F_{_T}(X),$ rather than their product $f(X)$ (as shown in Step (5)). However, the following lemma shows that both are equivalent. 
\begin{lemma}
\label{productgcd}
Let ${\cal U}=\left\{f_i:f_i \in \mathbb{F}[X], i=1,2,...,n\right\}$ be a collection of univariate polynomials with coefficients from some field $\mathbb{F}.$ A polynomial $g \in \mathbb{F}[X]$ is relatively prime with all the polynomials in $\cal U$ if and only if it is relatively prime with their product. 
\end{lemma}
\textit{Proof:} Appendix \ref{appendixproductgcd}.
\subsection{Fast algorithm for computing least degree coprime polynomial}
\label{subsecfastcoprime}
Algorithm \ref{alg:coprime} is a fast method to compute the least degree irreducible polynomial $g(X)$ among irreducible polynomials up to some degree $m$ that is coprime with $f(X).$ 
\begin{algorithm}
$(1)$ Let ${\cal P}=\left\{ X^{2^i}+X:i=1,2,...,m \right\}.$ 

$(2)$ \ForEach{$i=1,2,...,m$}
{
Calculate $r(X)=f(X)(\text{mod}~p_i(X)).$

\If{$r(X)$ is non-zero}
{
Break.
} 
}

$(3)$ Pick $p_j(X)$ as the first polynomial (i.e. least degree) for which $r(X)$ is non-zero. Note that every $p_i(X)\in \cal P$ is the product of all irreducible polynomials whose degree divides $i.$ Also, all irreducible polynomials of degree $i < j$ divide $f(X)$ as all $p_i(X)|f(X)$ for all $i<j.$ Therefore, at least one of the irreducible polynomials of degree $j$ is coprime with $f(X).$ 

$(4)$ Find one such polynomial $g(X)$ of degree $j$ which is coprime with $f(X)(mod~p_j(X))$ and therefore equivalently with $f(X)$ (Subsection \ref{subsec3b} gives a justification of this step).
\caption{Fast algorithm for computing $g(X)$\vspace{0.2cm}}
\label{alg:coprime}
\end{algorithm}
As a result, the key step (Step (5)) of Algorithm \ref{alg:necclowfieldsize} can be performed much faster than having to compute $g(X)$ by brute-force. Similarly, this fast algorithm also enables to quicken the key step (Step (4)) of Algorithm \ref{alg:construction} so that its overall complexity is reduced. 

Note that for using Algorithm \ref{alg:coprime} to implement Step (4) of Algorithm \ref{alg:construction}, we fix $m=\lceil log(N) \rceil$ as any polynomial $g(X)$ coprime with $f(X)$ is useful only if the degree of $g(X)$ is less than $\lceil log(N) \rceil +1,$ as only such a $g(X)$ can result in a network code using a smaller field than the one we started with. For the same reason, in using Algorithm \ref{alg:coprime} in conjunction with Algorithm \ref{alg:necclowfieldsize}, we choose $m=\lceil log(N|{\cal F}|)  \rceil.$ 
\subsection{Justification for Algorithm \ref{alg:coprime}}
\label{subsec3b}
The following lemma ensures that all polynomials which are found to be coprime with $f(X)$ by directly computing the gcd (or the remainder for irreducible polynomials) in the brute force method (as done in Algorithm \ref{alg:construction}), can also be found by running Algorithm \ref{alg:coprime}, using the set of polynomials $\cal P$ up to the appropriate degree.
\begin{lemma}
\label{lemmafindingg}
For some field $\mathbb{F},$ let $f,g\in\mathbb{F}[X]$ be two polynomials. Let $p \in \mathbb{F}[X]$ be such that $g|p.$  Then $g$ is relatively prime with $f$ if and only if $g$ is relatively prime with $f(\text{mod}~p).$
\end{lemma}
\textit{Proof:} Appendix \ref{appendixlemmafindingg}.
\subsection{Complexity of Algorithm \ref{alg:coprime}} 
The following proposition gives the complexity of Algorithm \ref{alg:coprime} for obtaining the coprime polynomial.
\begin{proposition}
\label{propcomplexitygcd}
The complexity of Algorithm \ref{alg:coprime} is at most $O(2^{2m})+O(mM),$ where $m=|{\cal P}|,$ and $M=deg(f(X)).$
\end{proposition}
\textit{Proof:} Appendix \ref{appendixpropcomplexitygcd}.
\begin{remark}
Note that the worst-case complexity of Algorithm \ref{alg:coprime} with $m=\lceil log(N) \rceil$ and $M=hN\lceil log(N)\rceil \Lambda$ (corresponding to values required for running Step(4) of Algorithm \ref{alg:construction}) is $O(N^2)+O(hN\Lambda(log(N))^2).$ This is clearly lesser than the worst-case complexity of finding the coprime polynomial $g(X)$ by brute-force, indicated in Section \ref{sec2}. Even if we test for coprimeness only for polynomials up to degree $\lceil log(N) \rceil,$ a brute-force execution of Step (4) of Algorithm \ref{alg:construction} would have a worst-case complexity of $O\left(N^2h\Lambda log(n)\right)$ (where $n=Nh\Lambda log(N)$), which is still greater than that of ours. 
\end{remark}
\subsection{Complexity of Algorithm \ref{alg:necclowfieldsize}} 
We now calculate the complexity of Algorithm \ref{alg:necclowfieldsize} (with Algorithm \ref{alg:coprime} used to implement its key step). The complexities of all the steps of Algorithm \ref{alg:necclowfieldsize} is given by Table \ref{tab1}, along with the references and reasoning for the mentioned complexities.

The only complexity calculations of Table \ref{tab1} which are not straightforward are the complexities involved in calculating the polynomial $g(X)$ coprime to $f(X)$ and in calculating the non-zero minor of the matrix $B_T^F.$ The complexity of calculating $g(X)$ can be calculated using Proposition \ref{propcomplexitygcd} using the values $m=\lceil N|{\cal F}| \rceil$ and $M=Nh|{\cal F}|\Lambda \lceil log\left(N|{\cal F}|\right)\rceil $. 

Now for calculating the non-zero minor of the matrix $B_T^F.$ There are $\left(\begin{array}{c}h \\k+m_{_T}^F \end{array}\right)$ such minors, and calculating each takes $O\left(\left(k+m_{_T}^F\right)^3\right)$ multiplications over $\mathbb{F}_q.$  As $\left(k+m_{_T}^F\right)$ can take values up to $h,$ clearly the function to be maximized is of the form $f(m)=\left(\begin{array}{c}h \\m \end{array}\right)m^3,$ for $m=0,1,...,h.$ Proposition \ref{thmcomplexity} gives the value of $m$ for which such a function is maximized, based on which the value in Table \ref{tab1} has been calculated. 
\begin{proposition}
\label{thmcomplexity}
For some positive integer $n,$ let $m$ be an integer such that $0\leq m \leq n.$ The function 
$
f(m)=
\left(\begin{array}{c}n \\m \end{array}\right)
m^3
$
is maximized at 
$
m = \left\{ 
\begin{array}{cc}
\left(\lceil\frac{n}{2}\rceil +1\right) & ~\text{if}~n \geq 2\\
1 & ~\text{if}~ n =1. 
\end{array}
\right.
$  
\end{proposition}
\textit{Proof:} Appendix \ref{appendixthmcomplexity}.
\begin{table*}
\caption{Complexity calculations for Algorithm \ref{alg:necclowfieldsize}}
\label{tab1}
\centering
\normalsize
\begin{tabular}{|c|c|c|c|}
\hline
\textbf{Step(s)} & \textbf{Complexity} & \textbf{Reasoning}\\
\hline
Algorithm \ref{alg:necc} & $A:=O\left(|{\cal F}|Nh\left(|{\cal E}||{\cal F}|N+|{\cal E}|+h+2\alpha\right)\right).$ & \cite{Mat} \\
\hline
Identifying non-zero minor of matrix $B_T^F$ & $B:=O\left(\left(
\begin{array}{c}
h \\
m
\end{array}
\right)m^3 \right),$ with $m=\left(\lceil\frac{h}{2}\rceil +1\right)$  & Theorem \ref{thmcomplexity}\\
\hline
Computing the non-zero minor (over $\mathbb{F}_2[X]$) of $B_T^F$ & $C:=O\left(h^4\Lambda log(N|{\cal F}|)\right)+O\left(\left(h\Lambda log\left(N|{\cal F}|\right)\right)^3\right) $ & \cite{EbF}\\
from a $(k+m_{_T}^F)$ square submatrix & & \\
\hline
Calculating $f(X)= \prod_{F\in {\cal F}}\prod_{T\in {\cal T}}f_{_T}^F(X).$ & $D:=O\left(alog(a)\right),$ where $a=Nh|{\cal F}|\Lambda log\left(N|{\cal F}|\right)$ & \cite{BoM}\\
\hline
Computing the coprime polynomial $g(X)$ & $E:=O\left( N^2|{\cal F}|^2\right)+O(Nh|{\cal F}|\Lambda log\left(N|{\cal F}|\right)^2).$ & Proposition \ref{propcomplexitygcd}\\
\hline
\hline
Total complexity & \multicolumn{2}{|c|}{$A + N|{\cal F}|(B + C + D) + E$} \\
\hline
\end{tabular}
\end{table*}
\section{Illustrative example - Network-Error Correction}
\label{sec5}
The performance of Algorithm \ref{alg:construction} (together with Algorithm \ref{alg:coprime}) for a network coding problem on a combination network is shown in Appendix \ref{appexm}. We now present a network-error correction example that uses Algorithm \ref{alg:necclowfieldsize} (with Algorithm \ref{alg:coprime}).
\begin{example}
\label{exm2}
Consider the network, with $18$ edges, shown in Fig. \ref{fig:necnetwork}. This network is from \cite{YaY}, in which a $1$ network-error correcting code meeting the network Singleton bound is given by brute-force construction for this network over $\mathbb{F}_4,$ which is the smallest possible field over which such a code exists. 
\begin{table*}
\small
\centering
\caption{Using Algorithm \ref{alg:necclowfieldsize} for the network in Fig. \ref{fig:necnetwork}} 
\begin{tabular}{|c|c|c|}
\hline
\textbf{Algorithm parameter} & \textbf{Network code defined by} $\boldsymbol{\cal A}$ &   \textbf{Network code defined by} $\boldsymbol{\cal B}$\\
\hline
Degree of $f(X)$, & & \\
the product of the $306$ determinant polynomials & $260$ & $978$ \\
\hline
$p(X)$: First $p_i(X)$ for which $f(X)(\text{mod}~p_i(X))$ is non-zero & $X^8+X$ & $X^4+X$\\
\hline
$f(X)(\text{mod}~p(X))$ & $X^7+X^6+X^3+X^2$ & $X^3+X$\\
\hline
$g(X)$: Least degree polynomial coprime to $f(X)$ & $X^3+X+1$ & $X^2+X+1$\\
\hline
$\left\{X_1,X_2,X_3\right\}$ after the algorithm & $\left\{\beta_8^1,\beta_8^3,\beta_8^3\right\}$ & $\left\{\beta_4,\beta_4,\beta_4\right\}$ \\
\hline
\end{tabular}
\label{tab3}
\end{table*}
According to the algorithm in \cite{Mat}, a $1$ network-error correcting code can be constructed deterministically if $q>2\left(\begin{array}{c}18 \\2 \end{array}\right)=306.$ In Fig. \ref{fig:necnetwork}, let the variable $X_1$ denote the encoding coefficient between edges $v_1\rightarrow v_4$ and $v_4\rightarrow v_6.$ Similarly, let the variable $X_2~(X_3)$ denote the local encoding coefficients between $v_2\rightarrow v_5~(v_6\rightarrow v_7)$ and $v_5\rightarrow v_8~(v_7\rightarrow v_9).$ 

\begin{figure}[htbp]
\centering
\includegraphics[totalheight=2.5in,width=3.2in]{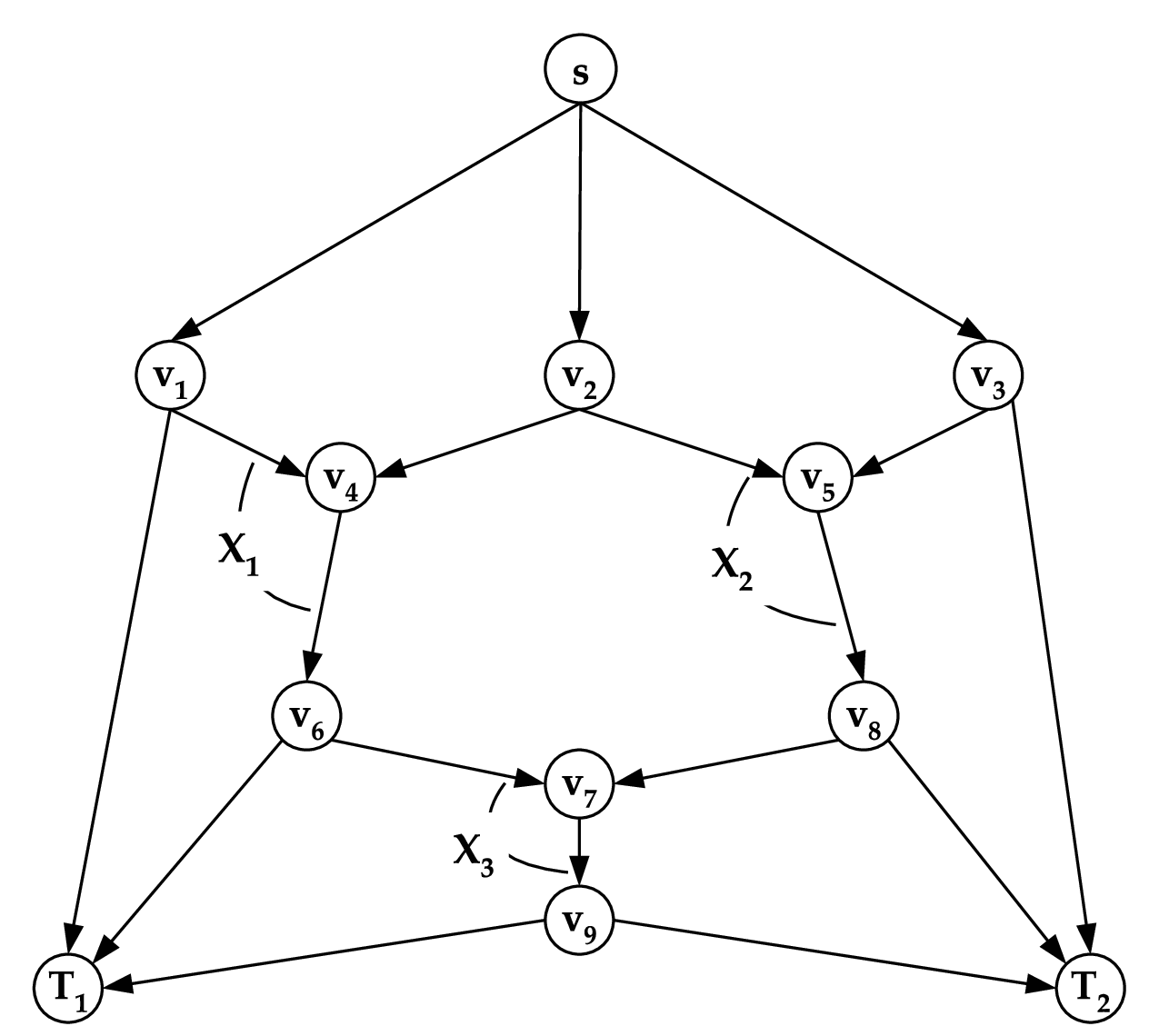}
\caption{Example network for network-error correction}	
\label{fig:necnetwork}	
\end{figure}

Let $q=2^9.$ Let ${\cal A}=\left\{\beta,\beta^{130},\beta^{130}\right\}$ and ${\cal B}=\left\{\beta^{132},\beta^{391},\beta^{391}\right\}$, where $\beta$ is a primitive element of $\mathbb{F}_{2^9}.$ Let $b_1(X)=X^9+X^4+1$ be the primitive polynomial of degree $9$ under consideration.

Consider two such $1$ network-error correcting codes obtained using Algorithm \ref{alg:necc} for the network of Fig. \ref{fig:necnetwork} as follows. Let $\cal A$ and $\cal B$ be two choices for the set $\left\{X_1,X_2,X_3\right\}$ with all the other local encoding coefficients being unity.  It can be verified that these two network codes can be used to transmit one error-free $\mathbb{F}_{2^9}$ symbol from the source to both sinks, as long as only single edge errors occur in the network. Table \ref{tab3} gives the results of running Algorithm \ref{alg:necclowfieldsize} for this network starting from these two codes, with $\beta_4$ and $\beta_8$ being the primitive elements of $\mathbb{F}_4$ and $\mathbb{F}_8$  respectively. 

Except for  $\left\{X_1,X_2,X_3\right\},$ all the other coding coefficients remain $1$ over the respective fields. It is seen from Table \ref{tab3} that the initial choice of the sets $\cal A$ and $\cal B$ for $\left\{X_1,X_2,X_3\right\}$ affects the complexity of the problem (i.e., degree of $f(X)$) and also the field size of the final network code. With $\cal B$, the resultant network-error correcting code is over $\mathbb{F}_4,$ exactly the one reported in \cite{YaY} by brute force 	construction. Also, for sink $T_1$ and $T_2,$ the value of $R_T(\alpha)$ can be computed to be $R_{T_1}(1)=R_{T_2}(1)=65.$ Thus the bound from \cite{GFZ} shown in (\ref{eqn15a}) for this network can be computed to be $q>130.$ The field size of the network-error correcting code found using our algorithm can therefore still be lesser than that of the bound in \cite{GFZ}.
\end{example}
\section{Concluding remarks}
\label{sec6}

As in the original paper \cite{EbF}, questions remain open about the designing of a code using the minimal field size. The hardness of calculating the minimal field size is reflected by the fact that the initial choice of the network code and the primitive polynomial of the field over which the initial code is defined (using which the local encoding coefficients are represented as polynomials) control the resultant field size after the algorithm. These issues are illustrated by the examples in Section \ref{sec5} and Appendix \ref{appexm}. However, it would be interesting to see if guarantees on the reduction of the field size can be given.

\appendices
\section{Proof of Lemma \ref{productgcd}}
\label{appendixproductgcd}
\begin{IEEEproof}
\textit{If part:} If $g$ is relatively prime with the product of all the polynomials in $\cal U,$ then there exist polynomials $a,b \in \mathbb{F}[X]$ such that 
\begin{align}
\label{eqn102}
a\left(\prod_{i=1}^{n}f_i\right)+bg=1.
\end{align}
For each $j \in {1,2,...,n},$ we can rewrite (\ref{eqn102}) as
\begin{align*}
\left(a\prod_{i=1,i\neq j}^{n}f_i\right)f_j+bg=1,
\end{align*}
which implies that $g$ is coprime with each $f_j\in {\cal U}.$

\textit{Only if part:} Suppose $g$ is relatively prime with all the polynomials in $\cal U.$ Then, for each $j \in {1,2,...,n},$ we can find polynomials $a_j$ and $b_j$ such that, $a_jf_j +b_jg=1.$ In particular, 
\begin{align}
\label{eqn103}
a_1f_1 +b_1g=1,\\
\label{eqn104}
a_2f_2 +b_2g=1.
\end{align}
Using (\ref{eqn104}) in (\ref{eqn103}), 
\begin{align*}
1&=a_1f_1(a_2f_2 +b_2g) +b_1g\\
&=(a_1a_2)f_1f_2+(a_1f_1b_2+b_1)g.
\end{align*}
Thus, $g$ is relatively prime with $f_1f_2.$ Continuing with the same argument, it is clear that $g$ is relatively prime with $\prod_{i=1}^{n}f_i.$
\end{IEEEproof}
\section{Proof of Lemma \ref{lemmafindingg}}
\label{appendixlemmafindingg}
\begin{IEEEproof}
Let $f=qp+r$ for the appropriate quotient and remainder polynomials $q,r \in \mathbb{F}[X]$ with $deg(r) < deg(p).$ Also, as $g|p,$ let $p=hg,$ for the appropriate $h\in \mathbb{F}[X].$

\textit{If part:}
As $r=f(mod~p)$ and $g$ are relatively prime with each other, we can obtain polynomials $a',b' \in \mathbb{F}[X]$ such that $a'r+b'g=1.$ Then, we must have
\begin{align*}
1&=a'(f-qp)+b'g\\
&=a'(f-qhg)+b'g \\
&=a'f+(b'-a'qh)g.
\end{align*}
Thus $f$ and $g$ must be coprime with each other. 

\textit{Only If part:}
Now assume that $f$ and $g$ are coprime with each other. This means we can obtain polynomials $a,b \in \mathbb{F}[X]$ such that $af+bg=1.$ Then, 
\begin{align*}
1&=a(qp+r)+bg\\
&=a(qhg+r)+bg \\
&=ar+(aqh+b)g,
\end{align*}
which means that $g$ and $r$ are coprime with each other, hence proving the lemma.
\end{IEEEproof}
\section{Proof of Proposition \ref{propcomplexitygcd}}
\label{appendixpropcomplexitygcd}
Towards proving Proposition \ref{propcomplexitygcd}, we first prove the following lemma. 
\begin{lemma}
\label{lemmafindinggcomplexity}
Let $f,p \in \mathbb{F}_2[X],$ be such that $deg(f)=M$ and $p=X^n+X,$ for some non-negative integers $M$ and $n.$ The polynomial $f(\text{mod}~p)$ can be calculated using at most $O(M)$ bit additions.
\end{lemma}
\begin{IEEEproof}
Let $f=\sum_{i=0}^{M}f_iX^i.$ We arrange the coefficients of $f$ as follows.
\[
\begin{array}{ccccccc}
f_0 & f_1 & f_2 &...&...&...&f_{n-1}\\
~&  f_n & f_{n+1}&...&...&...& f_{2n-2}\\
~&  f_{2n-1} & f_{n+1}&...& ...& ...&f_{3n-3}\\
~& . & . & ... & ... &... & . \\
~& . & . & ... & ... & ...&. \\
~& f_{an-a+1} & ... & f_M &0 & ... &~~~~0~~~, 
\end{array}
\]
where $a$ is the largest positive integer such that $an-a+1\leq M.$

Now, note that calculating the polynomial $f(\text{mod}~p)$, is equivalent to adding up the rows of the arrangement, while retaining the coefficient $f_0$ as it is. There are $\left\lceil\frac{M}{n-1} \right\rceil$ rows in the arrangement, and adding any two rows requires at most $n-1$ additions. Thus, the total number of bit additions is $O(M).$
\end{IEEEproof}

We are now ready to prove Proposition \ref{propcomplexitygcd}.
\subsection{Proof of Proposition \ref{propcomplexitygcd}}
\begin{IEEEproof}
The worst-case for Algorithm \ref{alg:coprime} would be $j=m.$ By Lemma \ref{lemmafindinggcomplexity}, computing $f(X)(\text{mod}~p_i(X))$ for some $p_i(X) \in \cal P$ takes at most $M$ operations. As there are $m$ such $p_i(X)$s, evaluating the remainders $r(X)$s costs  $Mm$ operations at most. Let 
\[
f(X)(\text{mod}~p_j(X))=\tilde{f}(X)
\]
be the non-zero polynomial of degree at most $2^m=2^j.$

Now, we have to determine the complexity in obtaining the polynomial of degree $m$ which is coprime with $f(X)$ (or equivalently with $\tilde{f}(X)$). 

There are approximately $\frac{2^j}{j}$ irreducible polynomials of order $j.$ It is known (see \cite{BoM}, for example) that for any two polynomials $p(X)$ and $q(X)$ (with degree $w$ of $p(X)$  larger than degree of $q(X)$), the complexity of dividing $p(X)$ by $q(X)$ (or equivalently, calculating $p(X)(\text{mod}~q(X))$) is $wlog(w).$ Thus, the complexity of dividing $\tilde{f}(X)$ by every possible irreducible polynomial of degree $j=m$ is at most $\frac{2^m}{m}O\left(2^mlog(2^m)\right)=O(2^{2m}).$

Thus, the total complexity for finding the least degree polynomial $g(X)$ coprime with $f(X)$ (which is assured of having a coprime factor of degree $m+1$) is at most $O(2^{2m})+O(Mm).$
\end{IEEEproof}

\begin{table*}
\centering
\small
\caption{$_6C_3$ network - Algorithm \ref{alg:construction} (together with Algorithm \ref{alg:coprime})}
\begin{tabular}{|c|c|c|c|c|}
\hline
\textbf{Algorithm parameter} & \multicolumn{2}{|c|}{ \textbf{Global encoding vectors} $\boldsymbol{\cal A}$} &  \multicolumn{2}{|c|}{ \textbf{Global encoding vectors} $\boldsymbol{\cal B}$}\\
\cline{2-5}
 & \textbf{Prim. poly. $b_1(X)$} &\textbf{Prim. poly. $b_2(X)$} &\textbf{Prim. poly. $b_1(X)$} &\textbf{Prim. poly. $b_2(X)$} \\
\hline
Degree of $f(X),$ the product of the & & & & \\
$20$ determinant polynomials & $20$ & $40$ & $30$ & $55$ \\
\hline
$p(X)$: First $p_i(X)$ for which  & & & & None of the form \\
$f(X)(\text{mod}~p_i(X))$ is non-zero & $X^4+X$ & $X^8+X$ & $X^8+X$ & $X^{2^i}+X,$ for $i\leq 4$\\ 
\hline
$f(X)(\text{mod}~p(X))$ &  $X^2+X$ & $X^7+X^6+X^3+X$ & $X^7+X^6+X^5+X^2$ & Not applicable\\
\hline
$g(X)$: Least degree & & & & \\
polynomial coprime to $f(X)$ & $X^2+X+1$ & $X^3+X+1$ & $X^3+X+1$ & Not applicable\\
\hline
Resultant network code & 
$
\begin{array}{c}
\left[
\begin{array}{c}
1 \\
0\\
0
\end{array}
\right]
\left[
\begin{array}{c}
0\\
1\\
0
\end{array}
\right]\\\\
\left[
\begin{array}{c}
0 \\
0\\
1
\end{array}
\right]
\left[
\begin{array}{c}
1\\
1\\
1
\end{array}
\right]\\\\
\left[
\begin{array}{c}
1 \\
\beta_4\\
\beta_4^2
\end{array}
\right]
\left[
\begin{array}{c}
 1 \\
\beta_4^2\\
\beta_4
\end{array}
\right]
\end{array}
$ 
&
$
\begin{array}{c}
\left[
\begin{array}{c}
1 \\
0\\
0
\end{array}
\right]
\left[
\begin{array}{c}
0\\
1\\
0
\end{array}
\right]\\\\
\left[
\begin{array}{c}
0 \\
0\\
1
\end{array}
\right]
\left[
\begin{array}{c}
1\\
1\\
1
\end{array}
\right]\\\\
\left[
\begin{array}{c}
1 \\
\beta_8\\
\beta_8^4
\end{array}
\right]
\left[
\begin{array}{c}
 1 \\
\beta_8^4\\
\beta_8^2
\end{array}
\right]
\end{array}
$
&
$
\begin{array}{c}
\left[
\begin{array}{c}
1 \\
0\\
0
\end{array}
\right]
\left[
\begin{array}{c}
0\\
1\\
0
\end{array}
\right]\\\\
\left[
\begin{array}{c}
0 \\
0\\
1
\end{array}
\right]
\left[
\begin{array}{c}
1\\
1\\
1
\end{array}
\right]\\\\
\left[
\begin{array}{c}
1 \\
\beta_8\\
\beta_8^3
\end{array}
\right]
\left[
\begin{array}{c}
 1 \\
\beta_8^3\\
\beta_8^6
\end{array}
\right]
\end{array} 
$
& Not applicable \\
\hline
\end{tabular}
\label{tab2}
\end{table*}
\section{Proof of Proposition \ref{thmcomplexity}}
\label{appendixthmcomplexity}
\begin{IEEEproof}
The statement of the theorem is easy to verify for $n=1.$ Therefore, let $n \geq 2.$ Let $g(k)=f(k)-f(k+1),$ for some $k,$ such that $0\leq k \leq n-1.$ Then, 
\begin{align*}
g(k)&=\left(\begin{array}{c}n \\k \end{array}\right)k^3-\left(\begin{array}{c}n \\k+1 \end{array}\right)(k+1)^3 \\
&=\left(\begin{array}{c}n \\k \end{array}\right)\left(k^3 - \frac{(n-k)}{(k+1)}(k+1)^3\right)\\
&=\left(\begin{array}{c}n \\k \end{array}\right)\left(2k^3+k^2(2-n)+k(1-2n)-n\right)\\
&=\left(\begin{array}{c}n \\k \end{array}\right)\tilde{g}(k),
\end{align*}
where $\tilde{g}(k)=\left(2k^3+k^2(2-n)+k(1-2n)-n\right).$ Proving the statement of the theorem is then equivalent to showing that both of the following two statements are true, which we shall do separately for even and odd values of $n$.
\begin{itemize}
\item $\tilde{g}(k) < 0$ for all integers $0 \leq k \leq \lceil\frac{n}{2} \rceil.$
\item $\tilde{g}(k) > 0$ for all integers $\lceil\frac{n}{2} \rceil + 1 \leq k \leq n-1.$ 
\end{itemize}
\textit{Case-A} ($n$ \textit{is even}): 
Let $k = \frac{n}{2} + i,$ for some integer $i$ such that $-\frac{n}{2} \leq i \leq \frac{n}{2}-1.$ Then,
\begin{align}
\nonumber
\tilde{g}(k)& = 2k^3+k^2(2-2k+2i)+k(1-4k+4i)-2k+2i\\
\label{eqn100}
\tilde{g}(k)&=k^2(-2+2i) + k(4i-1)+2i.
\end{align}

For $-\frac{n}{2} \leq i \leq 0,$ it is clear from (\ref{eqn100}) that $\tilde{g}(k) < 0.$ If $1 \leq i \leq \left(\frac{n}{2}-1\right),$ it is clear that $\tilde{g}(k) > 0.$ Thus, for even values of $n$, the theorem is proved.\\
\textit{Case-B} ($n$ \textit{is odd}): 
Let $k = \lceil\frac{n}{2} \rceil + i = \left(\frac{n+1}{2}\right) + i,$ for some integer $i$ such that $-\left(\frac{n+1}{2}\right) \leq i \leq \left(\frac{n-3}{2}\right).$ Then,
\begin{align}
\nonumber
\tilde{g}(k)& = 2k^3\hspace{-0.7cm}&+~&k^2(2-2k+2i+1)\\
\nonumber
&\hspace{-0.6cm}&+~&k(1-4k+4i+2)-2k+2i+1\\
\label{eqn101}
\tilde{g}(k)&=&&\hspace{-1cm}k^2(-1+2i) + k(1+4i)+2i+1.
\end{align}

Now, for $i=0,$ $k= \left(\frac{n+1}{2}\right) \geq 2$ (as $n \geq 2$ and is odd). Hence, $\tilde{g}(k)=-k^2+k+1 < 0$ for $i=0.$ If $-\left(\frac{n+1}{2}\right) \leq i < 0,$ then by (\ref{eqn101}), it is clear that $\tilde{g}(k) < 0.$ Thus for all $-\left(\frac{n+1}{2}\right) \leq i \leq 0,$ $\tilde{g}(k) < 0.$

For $1 \leq i \leq \left(\frac{n-3}{2}\right),$ again by (\ref{eqn101}), it is clear that $\tilde{g}(k) > 0,$ and thus the theorem holds for odd values of $n.$ This completes the proof.
\end{IEEEproof}
\section{Example - Network coding}
\label{appexm}
\begin{figure}[htbp]
\centering
\includegraphics[totalheight=2.3in,width=3.3in]{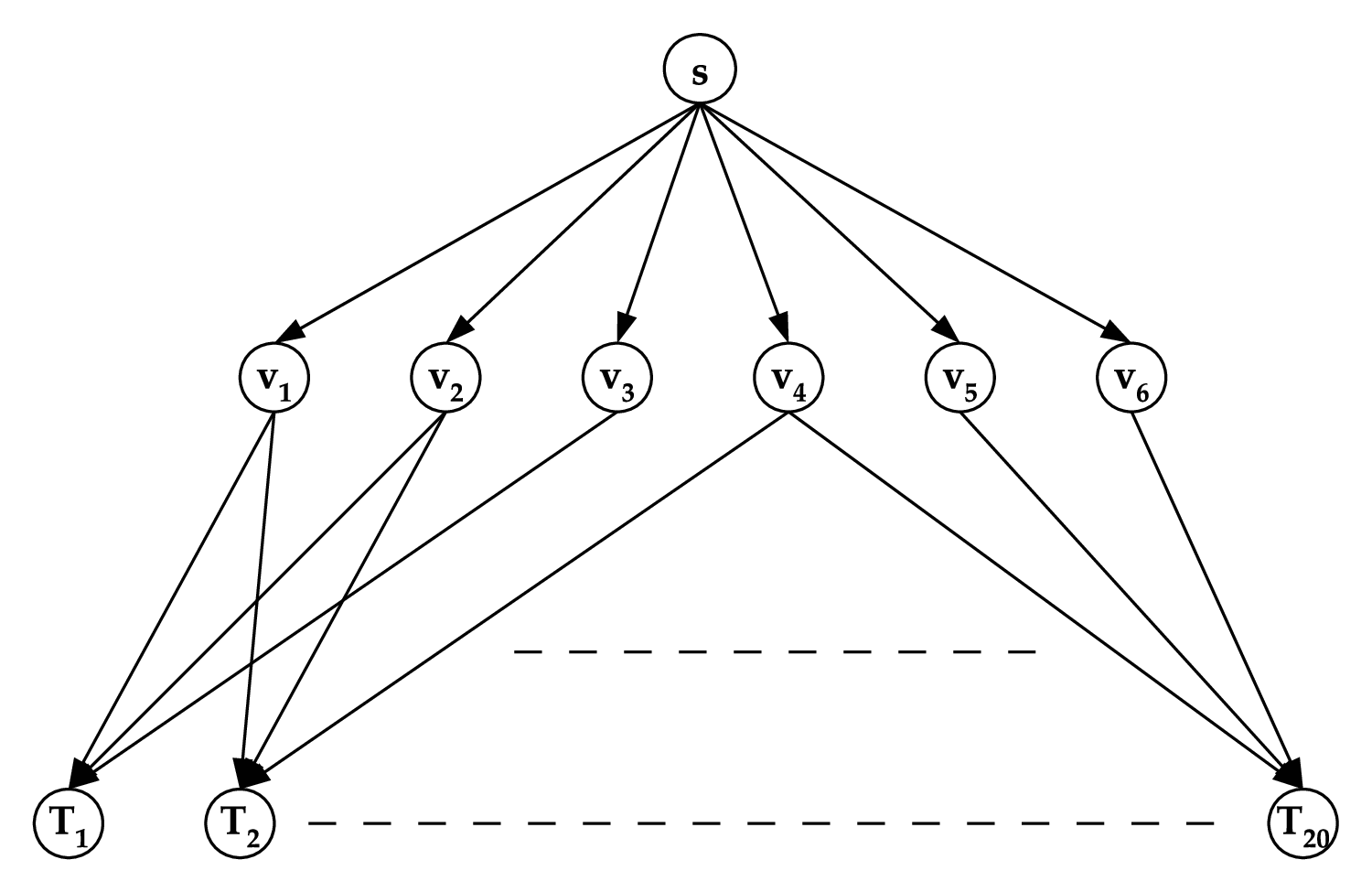}
\caption{$_6C_3$ network with $20$ sinks}	
\label{fig:6C3Network}	
\end{figure}
\begin{example}
\label{exm1}
Consider the $\left(\begin{array}{c}6 \\3 \end{array}\right)$ network shown in Fig. \ref{fig:6C3Network}. This network has $20$ sinks, each of which has $3$ incoming edges from some $3$-combination of the $6$ intermediate nodes, thus the mincut $h$ being $3$. Using the methods in \cite{KoM,JSCEEJT,Har}, a $3$-dimensional network code can be constructed for this network as long as the field size $q>20.$ Let $q=2^5.$ Consider the following sets of vectors in $\mathbb{F}_{32}^3,$ with $\beta$ being a primitive element of $\mathbb{F}_{32}.$
\[
\small
{\cal A}=\left\{
\begin{array}{c}
\left[
\begin{array}{c}
1 \\
0\\
0
\end{array}
\right],
\left[
\begin{array}{c}
0\\
1\\
0
\end{array}
\right],
\left[
\begin{array}{c}
0 \\
0\\
1
\end{array}
\right],\\\\

\left[
\begin{array}{c}
1\\
1\\
1
\end{array}
\right],
\left[
\begin{array}{c}
1 \\
\beta\\
\beta^{18}
\end{array}
\right],
\left[
\begin{array}{c}
 1 \\
\beta^{18}\\
\beta^{5}
\end{array}
\right]
\end{array}
\right\},
\]
\[
\small
{\cal B}=\left\{
\begin{array}{c}
\left[
\begin{array}{c}
1 \\
0\\
0
\end{array}
\right],
\left[
\begin{array}{c}
0\\
1\\
0
\end{array}
\right],
\left[
\begin{array}{c}
0 \\
0\\
1
\end{array}
\right],\\\\

\left[
\begin{array}{c}
1\\
1\\
\beta^6
\end{array}
\right],
\left[
\begin{array}{c}
1 \\
\beta\\
\beta^{18}
\end{array}
\right],
\left[
\begin{array}{c}
 1 \\
\beta^{18}\\
\beta^{5}
\end{array}
\right]
\end{array}
\right\}.
\]

Let $b_1(X)=X^5+X^2+1$ and $b_2(X)=X^5+X^3+X^2+X+1,$ both of them being primitive polynomials of degree $5.$ Note that $\cal A$ and $\cal B$ are valid choices (using either $b_1(X)$ or $b_2(X)$ as the primitive) for the global encoding vectors of the $6$ outgoing edges from the source, representing deterministic network coding solutions for a $3$-dimensional network code for this network. We assume that the intermediate nodes simply forward the incoming symbols to their outgoing edges, i.e., their local encoding coefficients are all $1.$

Table \ref{tab2} illustrates the results obtained with the execution of Algorithm \ref{alg:construction}, with Algorithm \ref{alg:coprime} being used to compute the coprime polynomial for this network with the original deterministic solutions being $\cal A$ or $\cal B,$ with $b_1(X)$ and $b_2(X)$ as the primitive polynomial of $\mathbb{F}_{32}$. 
%
The solutions (global encoding vectors of the $6$ edges from the source) obtained for the $\left(\begin{array}{c}6 \\3 \end{array}\right)$ network, after the modulo operations of the individual coding coefficients using the polynomial $g(X),$ are also shown in Table \ref{tab2}. It can be checked that both of these sets of vectors are valid network coding solutions for a $3$-dimensional network code for the $\left(\begin{array}{c}6 \\3 \end{array}\right)$ network. 

It is seen that for the set $\cal A$ being the choice of the network code in the first step of Algorithm \ref{alg:construction} and with $b_1(X)$ being the primitive polynomial, the final coprime polynomial has degree $2$ and thus resulting in a code $\mathbb{F}_4,$ which is in fact the smallest possible field for which a solution exists for this network. For $\cal B$ with the primitive polynomial $b_2(X),$ no solutions are found using characteristic two finite fields of cardinality less than $32.$
\end{example}

\end{document}